\documentclass[a4paper,english,numberwithinsect]{eurocg22}

\usepackage{microtype} 
\usepackage[ruled, vlined, english]{algorithm2e}
\usepackage{mathtools}

\DeclareMathOperator{\frechet}{d_{\mathcal{F}}}

\DeclareMathOperator{\sfrechet}{d_{\mathcal{S}}}
\DeclareMathOperator{\ofrechet}{d_{\mathcal{O}}}
\DeclareMathOperator{\tfrechet}{d_{\mathcal{T}}}

\DeclarePairedDelimiter\abs{\lvert}{\rvert}
\DeclarePairedDelimiter\norm{\lVert}{\rVert}%

\makeatletter
\let\oldabs\abs
\def\abs{\@ifstar{\oldabs}{\oldabs*}}
\let\oldnorm\norm
\def\norm{\@ifstar{\oldnorm}{\oldnorm*}}
\makeatother

\title{The $ k $-outlier Fréchet distance \footnote{The work was supported by the PhD School "SecHuman - Security for Humans in Cyberspace" by the federal state of NRW.}}
\titlerunning{The $ k $-outlier Fréchet distance}

\author[1]{Maike Buchin}
\author[1]{Lukas Plätz}
\affil[1]{Faculty of Computer Science, Ruhr-Universität Bochum\\
	\texttt{ \{maike.buchin, lukas.plaetz\}@rub.de}}
\authorrunning{M. Buchin and L. Plätz}

\begin{document}
	
	\maketitle
	
	\begin{abstract}
		The Fréchet distance is a popular metric for curves; however, its bottleneck character is a disadvantage in many applications.
		Here we introduce two variants of the Fréchet distance to cope with this problem and expand the work on shortcut Fréchet distances.
		We present an efficient algorithm for computing the new distance measure.
	\end{abstract}
	
	\section{Introduction}
	The analysis of curves is a growing field of study in computational geometry. The Fréchet distance is a popular metric between two curves. However, working with real-life data brings errors in measurement with it; and the Fréchet distance is quite sensitive to such outliers as it is defined as the largest distance in a minimal correspondence between the two curves. Hence outliers in the data may determine the Fréchet distance for the entire curves.
	
	There are three different approaches to handle such error: the partially measured \cite{buchin2009exact}, the averaged \cite{buchin2007computability, maheshwari2018approximating} and the shortcut Fréchet distance \cite{driemel2013jaywalking}. In this paper, we want to focus on the last one. 
	Driemel and Har-Peled introduced the shortcut Fréchet distance. It allows replacing parts of one curve through straight line segments, called shortcuts. Deciding the shortcut Frechet distance between two curves was shown NP-hard by Buchin, Driemel and Speckmann in \cite{buchin2014computing}. In the vertex restricted case, these line segments must start and end at vertexes. For this case Driemel and Har-Peled showed that the vertex restricted directed $k$-shortcut Fréchet distance, $k$ counting the number of shortcuts, can be approximated in near-linear time. Buchin, Driemel and Speckmann gave an algorithm to decide the vertex restricted shortcut distance exactly in $O(n^3 \log n)$ time. An open problem is to extend the shortcut distance to allow shortcuts on both curves without completely short-cutting both curves.
	Avraham et al. considered the discrete problem \cite{avraham2015discrete} and faced the same problem for the two-sided case. To resolve it, they forbid simultaneous movement on the curves.
	In \cite{alt2003matching} Alt et al. discuss first how to define and compute the Fréchet distance between a curve and a graph, and then between two graphs. It introduces a free space surface for these, which is similar in spirit to the outlier free space we define here.
		
	\begin{table}[]
	    \centering
	    \begin{tabular}{c|c|c}
	        one-sided & shortcut & $k$-outlier \\
	        \hline
	        discrete & $O((n+m)^{6/5+\epsilon})$ \cite{avraham2015discrete} & naive $ O(nmk  \log n )$ \\
	        continuous & NP-Hard \cite{buchin2014computing} | vertex restricted $ O(n^5 \log n) $ \cite{driemel2013realistic} & naive $ O((n^2mk + nmk^2) \log n) $ \\
	        \hline
	        \hline
	        two-sided & shortcut & $k$-outlier \\
	        \hline
	        discrete & $O((m^{2/3}n^{2/3}+m+n) \log^3 (m + n)) $ \cite{avraham2015discrete} & naive $ O(n m k^2 \log n)$  \\
	        continuous & -- & $ O((n^2mk + nmk^3) \log n) $ [Thm. \ref{thm:main}]
	    \end{tabular}
	    \caption{Computational complexity of the shortcut and outlier computation problem}
	    \label{tab:comparison}
	\end{table}
	
	Some results on curve simplification are special cases of our problem in the sense that the directed outlier distance of a curve to itself is considered here.
	In \cite{imai1988polygonal} Imai and Iri for simplifying a curve introduce an associated graph to the curve. This is also called the shortcut graph of a curve, which we will also later use in our algorithm.
	Bringmann and Chaudhury \cite{bringmann2019polyline} considered curve simplification with vertex restricted shortcuts. They showed that this simplification needs $ O(n^3) $ time. They also showed that for $ L_p $ with $ p \neq 2, \infty $ this cannot be done in $ O(n^{3- \varepsilon}) $ for all $ \varepsilon > 0 $ unless $ \forall \forall \exists \! - \! \text{OV}$ hypothesis fails. Kerkhof, Kostitsyna, Löffler, Mirzanezhad and Wenk \cite{vandekerkhof2019global} achieved the same runtime for their simplification algorithm. 
	
	We will address the two open problems of allowing shortcuts on both curves and taking into account the length of the shortcuts. For this, we present the $k$-outlier distances. These distances allow ignoring $k$ outliers on one or both curves and computing the optimal Fréchet distance given the number of vertices to leave out. Switching from counting shortcuts to counting vertices makes it possible to compute a symmetrical shortcut distance. Further we allow other starting and ending points, hence it can also be seen as a partial Fréchet distance.
	Table~\ref{tab:comparison} compares our result to previous results and naive algorithms.

	\section{continues \textit{k}-outlier Fréchet distance}
	\subsection{Outlier free space cell}
	We start by defining curves, shortcuts and the Fréchet distance.
	\begin{definition}
		Let $ X = \langle p_0, p_1, \dots, p_n \rangle $ be a polygonal curve. We consider $X$ as a continuous map $ X \colon [0, n] \rightarrow \mathbb{R}^d $, where $ X(i) = p_i $ for $ i \in \mathbb{N} $, and the $i$-th edge is linearly parametrized as $X (i + \lambda) = (1 - \lambda)p_i + \lambda p_{i+1} $. We denote the \textbf{shortcut} $\langle p_i, p_a \rangle $ for $ i < a$ as the straight line segment connecting the points. 
		
		A \textbf{reparametrisation} $ (\sigma, \theta) $ of two curves $X$ and $Y$ is a pair of continuous non-decreasing surjective functions, where $ \sigma $ and $ \theta $ map from $[0, 1]$ to $[0, n]$ and $[0, m]$, respectively.
		The \textbf{Fréchet distance} between two polygonal curves $X$ and $Y$ 
		is the maximum distance attained by optimal reparameterisations, i.e. $\frechet(X, Y) := \inf_{(\sigma, \theta )}  \max_t \norm{X (\sigma (t)) - Y(\theta (t))}$. 
		A reparametrisation of maximum distance at most $ \varepsilon $ is called an {\bf $\varepsilon$-realization}.
	\end{definition}
	Now we can define our outlier distances.
	\begin{definition}[Outlier distance]
		A curve $ \overline{X} := \langle \overline{p_1}, \overline{p_2}, \dots, \overline{p_\ell} \rangle $ is in the set of \textbf{$ k $-outlier curves} $ C_k(X) $ if and only if the points $ \langle \overline{p_1}, \overline{p_2}, \dots, \overline{p_\ell} \rangle $
		are a subsequence of $ \langle p_1, p_2, \dots, p_n \rangle $ and $n-\ell \leq k$. 
		The \textbf{directed $ k $-outlier Fréchet distance} $ \ofrechet (k, X, Y) := \min_{\overline{Y} \in C_k(Y)} \frechet(X, \overline{Y}) $. 
		
		The set of \textbf{$ k $-outlier curve tuples} $ T_k(X, Y) := \bigcup_{i=0}^k C_i(X) \times C_{k - i}(Y) $ contains those curves, where the number of outliers of both curves is at most $ k $.
		The \textbf{undirected $ k $-outlier Fréchet distance} is $ \tfrechet (k, X, Y) := \min_{\left(\overline{X}, \overline{Y}\right) \in T_k(X, Y)} \frechet\left(\overline{X}, \overline{Y}\right) $.
	\end{definition}
	
	Note that we do not (necessarily) restrict to start or end $ \overline{X} $ with the same vertices as~$X$. We will later see that allowing this does not increase the computation time and that the directed case is a specialization of the undirected case. Hence in the remainder of the paper, the $ k $-outlier Fréchet distance will refer to the undirected case.

    \begin{remark}[Axioms of Metrics]
        The identity of indiscernibles does not hold for Fréchet curves because with the $ k $-outlier Fréchet distances we cannot distinguish between two curves with $ k $ different points.
        The undirected case is symmetrical because the set of $ k $-outlier curves tuples is symmetrical.
        The triangle inequality is not satisfied, this can be shown a counter-example. See figure \ref{fig:dreiecksungleichung}. For this example even $ \tfrechet(k, X, Y) + \tfrechet(k, Y, Z) \geq c\tfrechet(2k, X, Z)$ is not satisfied.
    \end{remark}
    
    \begin{figure}[htbp]
		\centering
		\includegraphics[width=0.7\linewidth]{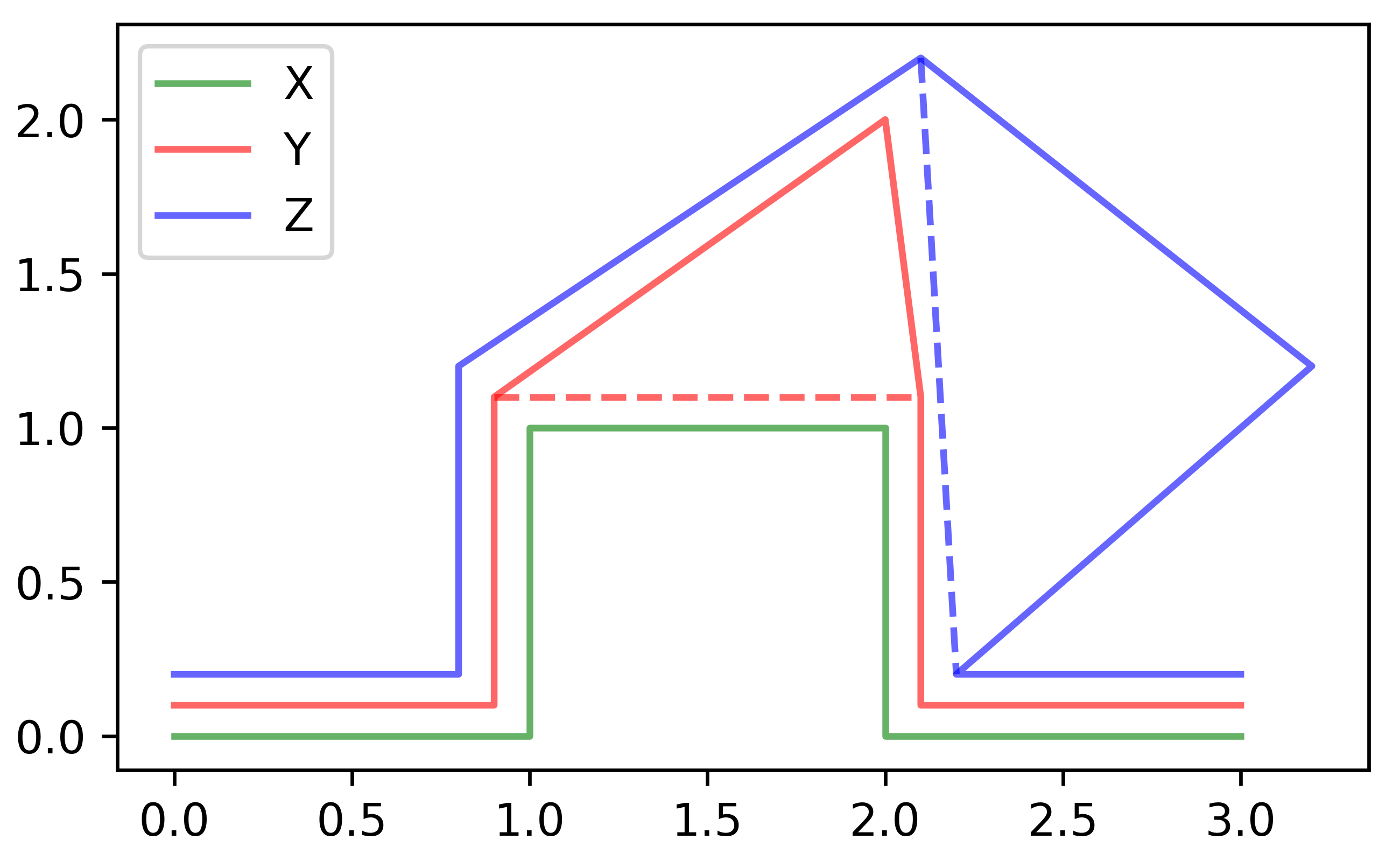}
		\caption{Example of Driemel and Har-Peled in \cite{driemel2013jaywalking}. $X$ and $Y$ are close under $ \sfrechet $ and $ \tfrechet $ for $k = 1$. The same is true for $Y$ and $Z$. For $X$ and $Z$ they are only close under $ \tfrechet $ for $k \geq 3$.}
		\label{fig:dreiecksungleichung}
	\end{figure}
    	
	With the shortcuts, we introduce new edges on our curves and with them, in turn, new cells in our free space by allowing these shortcuts. To distinguish the terminology of the classical and the outlier Fréchet distance, we add the word outlier to the new terms.
	
	We present a dynamic program to decide the $k$-outlier Fréchet distance. We introduce the height in the outlier free space to keep track of the number of outliers.
	We start by initialising the outlier free space with the starting points.
	Here we also allow omitting vertices at the beginning of both curves, with the number of omitted points defining the starting height.
	Then we compute the reachable intervals for each outlier cell with the starting point and the previously computed reachable outlier intervals. Here every cell defines how much height has to be added to the result.
	After the computation of an outlier cell, we have the classical reachable free space interval of that cell. But to compute the outlier interval we need all adjacent cells. We combine then all reachable intervals into the reachable outlier intervals.
	In the end, we check if any ending point is reachable. They are similar to the starting points but allow to use less than $k$ outliers, too.

	\begin{definition}[Outlier Cell]
		An \textbf{outlier cell} $ C[(i, j), (a, b), h] $ is defined as the classical free space cell of the edges $ \langle p_i, p_a \rangle $ and $ \langle p_j, p_b \rangle $ and with an added \textbf{height} $ h $. See \ref{fig:einzelnezelle} for a visualisation.
		With $R^h(C)$ and $R^v(C)$ we describe the \textbf{horizontal} and \textbf{vertical reachable interval} of an outlier cell $ C $.
		The \textbf{length} $ \text{L} $ of a shortcut $ \langle p_i, p_a \rangle $ with $ a > i $ is defined as $ a - i - 1 $. It counts the vertices skipped by the shortcut. For the special case $ a = i $ we set the length to $ 0 $. See \ref{fig:abkurzung} for an example.
		The \textbf{horizontal length} $ \text{L}_h(C[(i, j), (a, b), h]) := \text{L}(\langle p_i, p_a \rangle) $ and the \textbf{vertical length} $ \text{L}_v(C[(i, j), (a, b), h]) := \text{L}(\langle p_j, p_b \rangle) $.
		A \textbf{starting point} $ (i, j) $ of an outlier curve tuple is the pair of its first vertices. The height of that point in the outlier free space is simply the sum of both indices. An \textbf{ending point} $ (n - i, m - j)$ is a pair of the last vertices of an outlier curve tuple. The height of this point has to be below $ k - i - j $ to allow the last points to be left out. 
		Furthermore a \textbf{vertical} and \textbf{horizontal reachable outlier interval} $I[(i, j), (i, b), h] := \bigcup_{l \in [k+1]} R^v(C[(i, j), (i+l, b), h]) $ and $I[(i, j), (a, j), h] := \bigcup_{l \in [k+1]} R^h(C[(i, j), (a, j+l), h]) $.
		The \textbf{horizontal free space interval} $ F[(i, j), (a, j)] := \lbrace p \in [i, a] \times \lbrace j \rbrace \,|\, \norm{X(x_p) - Y(y_p)} \leq \varepsilon \rbrace $ only depends on the edge $ \langle p_i, p_a \rangle $ and the point $ p_j $. The similar is true for the \textbf{vertical free space interval}  $ F[(i, j), (i, b)] := \lbrace p \in  \lbrace i \rbrace \times [j, b] \,|\, \norm{X(x_p) - Y(y_p)} \leq \varepsilon \rbrace $ 
		We indicate by using an index twice that this is defined by a vertex rather than an edge. 
		An \textbf{outlier point} in the free space diagram is defined as $ P[(i, j), h] := I[(i, j), (a, j), h] \cap I[(i, j), (i, b), h] $.
	\end{definition}

	\begin{figure}[htbp]
		\centering
		\includegraphics[width=0.9\linewidth]{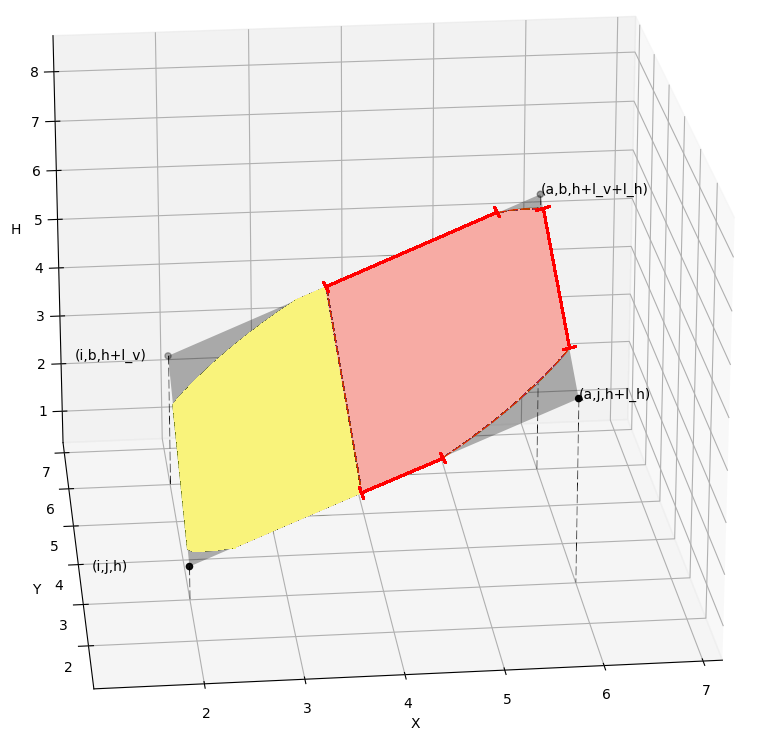}
		\caption{We visualise an outlier free space cell in 3-dimensional space. The four points of the cell $ C[(i, j), (a, b), h] $ are $ (i, j, h), (i, b, h + \text{L}_v), (a, j, h + \text{L}_h) $ and $(a, b, h + \text{L}_v + \text{L}_h) $. The lower left is at the same height as the cell. The upper left is lifted by the vertical length $ \text{L}_v $, the lower right by the horizontal length $ \text{L}_h $ and the upper right by both lengths. The reachable free space is shown in red.}
		\label{fig:einzelnezelle}
	\end{figure}
		
	\begin{figure}[htbp]
		\centering
		\includegraphics[width=0.7\linewidth]{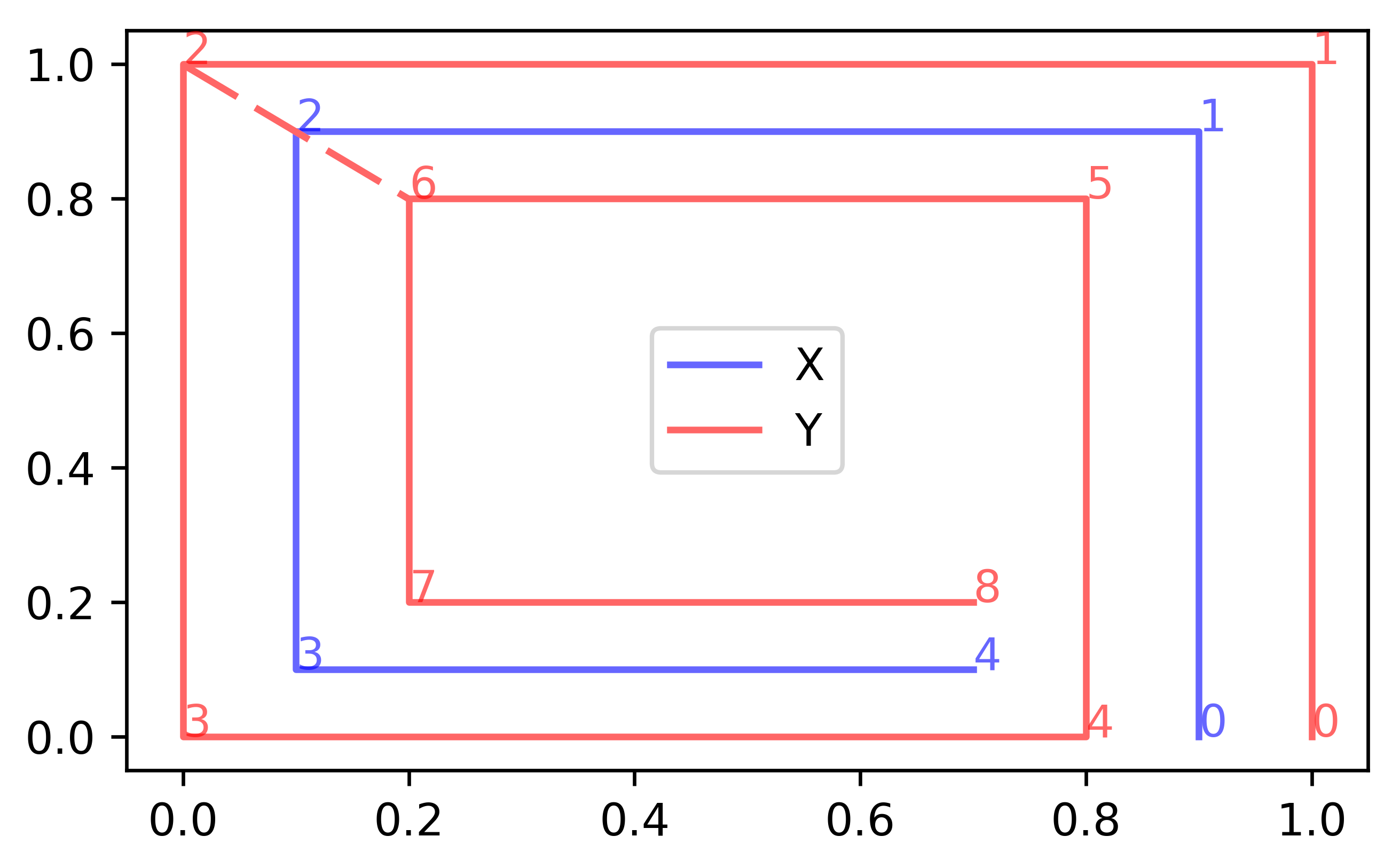}
		\caption{Two curves with a possible shortcut (dashed) on the red curve of length 3.}
		\label{fig:abkurzung}
	\end{figure}
	
	Figures~\ref{fig:00} to \ref{fig:12} show examples of two classical and an outlier free space diagram.
	
	\begin{figure}[htbp]
		\centering
		\includegraphics[width=0.9\linewidth]{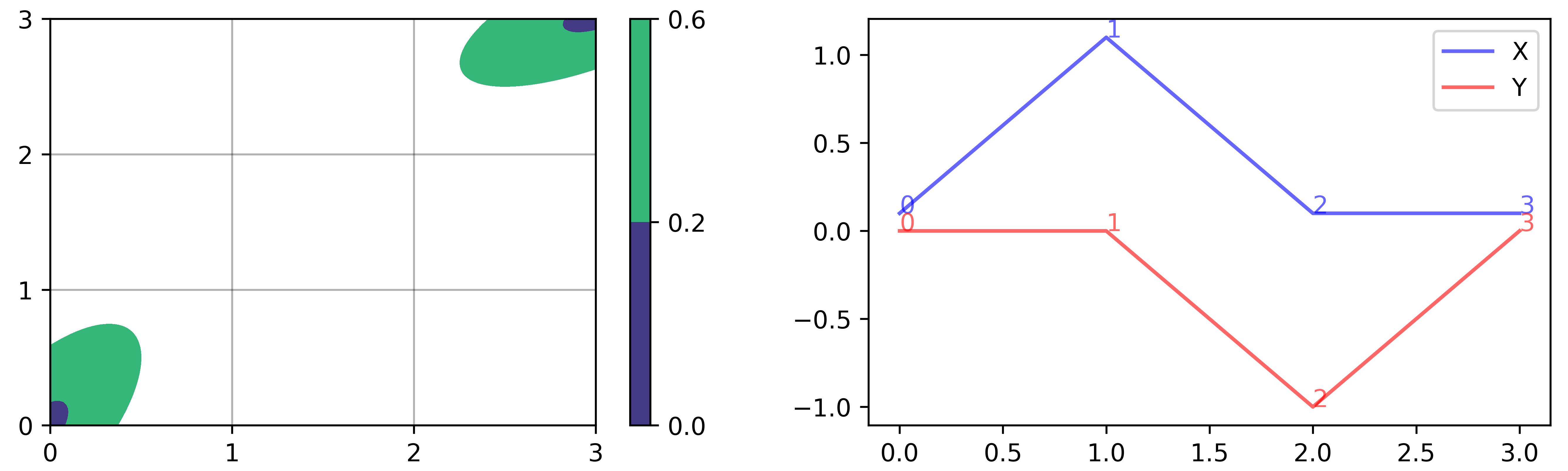}
		\caption{On the right side are the curves and on the left is the free space diagram for two values.}
		\label{fig:00}
		\includegraphics[width=0.9\linewidth]{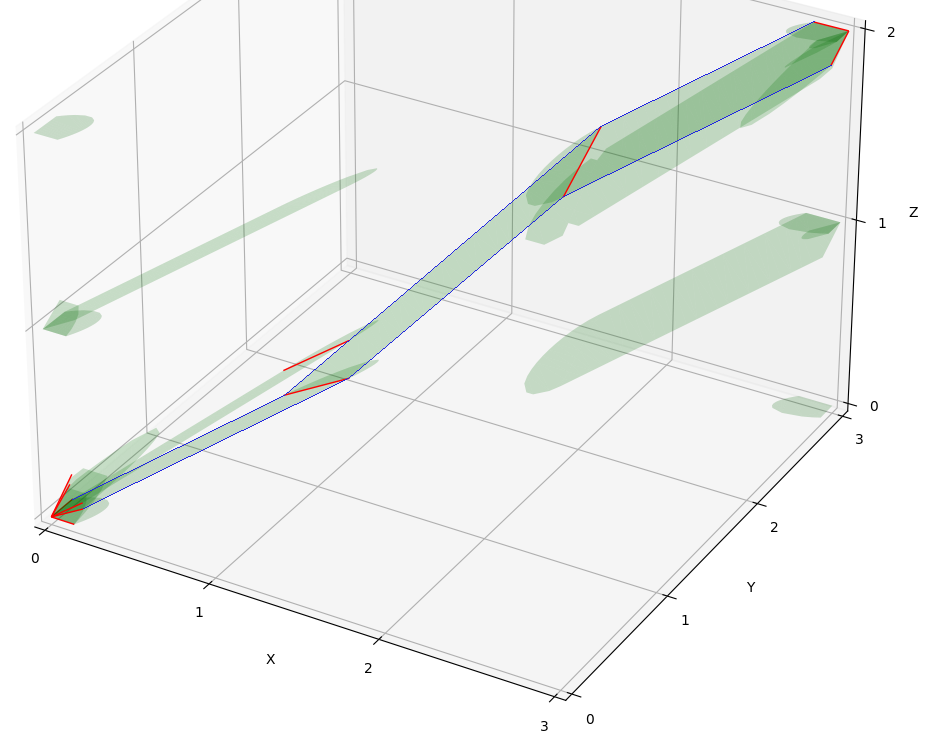}
		\caption{The outlier free space diagram for the same curves as in Figure~\ref{fig:00}. At the plane $z = 0$, we can see the classical free space diagram. The reachable free space intervals are drawn in red. The free space leading to an ending point of the curve is marked with a blue line.}
		\label{fig:outlierfreespace}
		\includegraphics[width=0.9\linewidth]{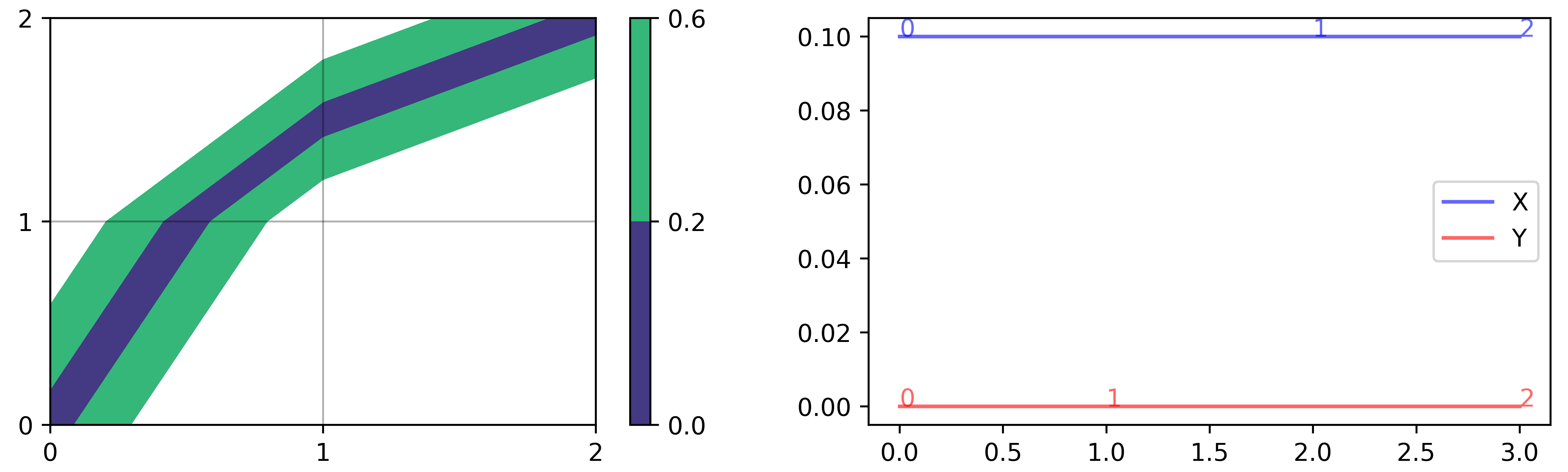}
		\caption{On the right side are the curves realizing the 2-outlier Fréchet distance of the curves in Figure~\ref{fig:00}. On the left is the classical free space of them.}
		\label{fig:12}
	\end{figure}

	\subsection{Algorithm}
	
	For an easier description, we will introduce the term predecessor.	
    Here and in the following we use $ [n] := \lbrace 0, \dots, n \rbrace $ for $ n \in \mathbb{N} $ as a shorthand.
	\begin{definition}
		The \textbf{predecessors} of a cell $ C[(i, j), (a, b), h] $ are the outlier intervals $ I[(i, j), (a, j), h] $ and $ I[(i, j), (i, b), h] $.
		The predecessors of an outlier interval $ I[(i, j), (i, b), h] $ and $ I[(i, j), (a, j), h] $ are the cells $ C[(i-l-1, j), (i, b), h-l] $ and $ C[(i, j - l - 1), (a, j), h-l] $ for $ l \in [k] $ respectively.
		The predecessors of a point $ P[(i, j), h] $ are the outlier intervals $ I[(i-l-1, j), (i, j), h-l] $ and $ I[(i, j - l - 1), (i, j), h-l] $ for $ l \in [k] $.
	\end{definition}
	In the following, we only consider curves with at least one edge and $ k \geq 1 $.
	With the terminology introduced in the previous section, we can decide the reachability of a point with the predecessors of the point. To compute an outlier interval applying two times its predecessors gives us a dependence only on $ O(k) $ intervals.
	
	We start Algorithm~\ref{Alg1} by initializing our reachable free space in the first loop and placing the starting points in the second loop. Then we can enforce the correct ordering of cells with a lexicographical ordering of the cells by their five identifying values. In this ordering, we compute the reachable outlier intervals in the third loop. After computing all cells we have to check in the last loop if any of the ending points are in the reachable free space below the needed height. 
	We test if an ending point is reachable by computing the reachable outlier interval ending with that point and accounting for the length of that interval.
	
		\begin{algorithm}[H]
		    \label{Alg1}
			\DontPrintSemicolon
			\caption{ Algorithm($ k, X, Y, \varepsilon $)}
			\tcp{Initialize Free Space}
			\ForEach(\tcp*[f]{Outlier Intervals}){$ (i, j, l, h) \in [n] \times [m] \times [k+1] \times [k] $}{
				Compute $ F[(i,j), (i + l, j)] $ and $ F[(i,j), (i, j + l)] $ \;
				Set $ I[(i, j), (i + l, j), h], I[(i, j), (i, j + l), h] $ and $ P[(i, j), h] $ as empty \;
			}
			\tcp{Adding Reachable Starting Points}
			\For(\tcp*[f]{Starting Points}){$ (i, j) \in [k]^2 $ with $ i + j \leq k $}{
				\lIf{ $ \norm{X(x_i) - Y(y_j) } \leq \varepsilon $ }{
				    Set $ (i, j) \in P[(i, j), i + j)] $
				    } 
			}
			\tcp{Computing the Reachable Space}
			\ForEach(\tcp*[f]{Cell $C[(i, j), (a, b), h] $}){$ (i, j, s, t, h) \in [n] \times [m] \times [k+1] \times [k+1] \times [k] $}{
				$ a = i + s, b = j + t $ and $ C = C[(i, j), (a, b), h]$ \;
				Update $ I[(i,b), (a, b), h+\text{L}_v] $ and $ I[(a, j), (a, b), h+\text{L}_h] $ with $ R^h(C) $ and $ R^v(C) $
				\;
			}
			\tcp{Testing Ending Points}
			\For(\tcp*[f]{Ending Points}){$  (i, j, h) \in [k]^3 $ with $ i + j \leq k, i + j + h \leq k  $}{
				\lIf{$ (n - i, m - j) \in P[(n - i, m - j), h)] $}{
					\Return True
				}
			}
			\Return False
		\end{algorithm}
	\begin{remark}
		There are several possibilities to alter the algorithm to compute different distances.
		The first would be to set $ k $ to $ 0 $. Then the algorithm is the same as algorithm one of Alt and Godau in \cite{alt1995computing} and the runtime collapses to $ O(nm) $. 
		
		The second would be to set $ a $ to $ i + 1 $ instead of iterating over it in the third loop. With this, we forbid the omission of a vertex on the first curve. The runtime of the third loop becomes $ O(nmk^2) $. 
		Setting also the $ i $ to $ 0 $ in the second and fourth loop we disallow any changes to the first curve and hence we get the directed outlier Fréchet distance.
		
		The third would be to also change the counting from points to shortcuts and only allow the starting and ending points $ (0, 0, 0)$ and $ (n, m, h) $ for $ h \in [k] $. Then the algorithm would decide the directed vertex restricted $ k $-shortcut Fréchet distance with runtime $ O(n^2mk) $. 
	\end{remark}
	For the correctness of the algorithm, we first show that reachability decides the $ k $-outlier distance.
	The proof uses two ideas. The first idea is to have for every pair of matched edges in the reparametrisation of an outlier curve tuple an outlier cell representing the free space of it. The second idea is to only use cells below or at the height $ k $. With these outlier cells and the outlier free space diagram we embedded the free space diagram of every outlier curve tuple in the outlier free space diagram and can count the number of outliers.
	\begin{lemma}
    	The $ k $-outlier distance is at most $ \varepsilon $ if and only if an ending point is reachable in the outlier free space below or on height $ k $.
	\end{lemma}
	\begin{proof}
		If the $ k $-outlier distance is at most $ \varepsilon $, there exists a limit of reparametrisation between an outlier curve tuple that realises the $ k $-outlier distance. Using only the vertices in the outlier curve tuple we get a classical free space diagram. The reparametrisation is a curve in this classical free space diagram. The sequence of cells given by that curve corresponds to a sequence of outlier cells in the outlier free space diagram. We start at a starting point in the outlier free space and end at an ending point. The curve through our sequence of outlier cells gives us the reachability of the ending point. The outlier tuple counts the number of outliers, so the height of the ending point is at most $ k $.  
		
		If an ending point is reachable with a height up to $ k $, the reachability is represented through a curve in the outlier free space from a starting point to that ending point. This curve gives us a sequence of outlier cells that correspond to classical cells which in turn give us a reparametrization of two curves. These two curves are an outlier tuple because the number of outliers has to be up to $ k $ and the sequence of curves gives us a sequence of edges for both curves which follow the definition of outlier curve tuples. Hence the reparametrisation of an outlier curve tuple proves that the $ k $-outlier distance is at most $ \varepsilon $.
	\end{proof}

	\begin{theorem}
		Algorithm~\ref{Alg1} decides correctly if the $ k $-outlier Fréchet distance is at most $ \varepsilon $. The runtime of Algorithm~\ref{Alg1} is in $ O(nmk^3) $. The space usage is in $ O(nmk^2) $.
	\end{theorem}
	\begin{proof}
		We will prove this by induction on the order of the outlier free space cells in the algorithm. 
		The first cell is a classical free space cell and only depends on the starting point $ (0,0,0) $. By initializing this point, we can compute the predecessors' intervals of this cell, so we only need to update the cell intervals accordingly. This is done as in \cite{alt1995computing} by Alt and Godau. There is no outlier to account for, so the heights of the other corners are $ 0 $.
		
		In the general case, the cell $ C[(i, j), (a, b), h] $ only depends on other predecessor intervals and an eventual starting point. But the intervals have lower lexicographical order and therefore are already computed. So we can compute the reachable outlier intervals of the current cell. With $ L_h(C) $ and $ L_v(C) $ we know the number of outliers on $ \langle p_i, p_a \rangle $ and $ \langle p_j, p_b \rangle $. We add the respective number to the height of outlier intervals of the cell. That completes the computation of that cell. 
		
		Now that the outlier reachable free space is computed we have to check if any of the ending points are reachable. An ending point $ (n - i, m - j) $ is $ i + j $ vertices apart from the ends of the curves, so we can only check heights $ h $ below $ k - i + j $. This can be done by checking if $ (n - i, m - j, h) $ is in the predecessor outlier interval $ I[(n - i - 1, m - j), (n - i, m - j), h] $.
		
		Now to the runtime.
		The first nested loops define the data structure we have to maintain. Every reachable interval needs $ O(1) $ space. So the whole algorithm needs $ O(nmk^2) $ space for the number of reachable free space intervals. 
		
		The second nested loop defines the running time of the algorithm. Computing the updates needs $ O(1) $ time. So the algorithm needs $ O(nmk^3) $ time because of the number of outlier free space cells. 
		Putting the starting point in the reachable free space needs $ O(k^2) $ time because we can place every starting point in constant time. 
		Checking if an ending point is in the reachable free space needs $ O(k^3) $ time because for every ending point $ k $ different heights have to be tested.
	\end{proof}

	\subsection{Critical values}
	
	We assume that $ n \geq m $.
	The three types of critical values, Alt and Godau introduced, also cover all cases of the outlier free space. The first case covers the different starting and ending points in the outlier free space, of which there are $ O(k^2) $ many. The second type describes the opening of free space intervals. A vertex and an edge define these. There are $ O(nmk) $ critical values of this type for the $ O(mk) $ edges and $ O(n) $ points. The last type is the monotonicity event or the opening of a passage. Here two vertices and an edge are the defining features. So we get $ O(n^2) $ for the vertices, $ O(mk) $ for the edge and $ O(n^2mk) $ in total of this type.
	
	\begin{theorem}
	    \label{thm:main}
	    The $ k $-outlier Fréchet distance can be computed in $ O((n^2mk + nmk^3) \log n)$.
	\end{theorem}
	\begin{proof}
	    The argument is the same as in \cite{alt1995computing} by Alt and Godau.
	    First, compute and sort the $ O(n^2mk) $ critical values and then use the decision algorithm with the runtime of $ O(nmk^3) $ in a binary search to find the $ k $-outlier Fréchet distance.
	\end{proof}

	\section{Conclusion}
	We introduced a new Fréchet variant to work with real-life data. We can decide the undirected $k$-outlier Fréchet distance of two curves of complexity $m,n$ in $ O(nmk^3) $ time and the directed distance in $ O(nmk^2) $ time.
	
	\bibliography{bibba}
\end{document}